\documentclass[a4paper,conference]{IEEEtran}

\usepackage{ifthen}
\usepackage[cmex10]{amsmath} 
\usepackage{macros}

\newtheorem{theorem}{Theorem}

\begin{document}
\title{One-Bit Quantizers for Fading Channels}

\author{%
\IEEEauthorblockN{Tobias Koch}
\IEEEauthorblockA{University of Cambridge\\
              Cambridge, CB2 1PZ, UK\\
              Email: tobi.koch@eng.cam.ac.uk}

\and

\IEEEauthorblockN{Amos Lapidoth}
\IEEEauthorblockA{ETH Zurich\\
              CH-8092 Zurich, Switzerland\\
              Email: lapidoth@isi.ee.ethz.ch}
}

\maketitle

\begin{abstract}
We study channel capacity when a one-bit quantizer is employed at the output of the discrete-time average-power-limited Rayleigh-fading channel. We focus on the low signal-to-noise ratio regime, where communication at very low spectral efficiencies takes place, as in Spread Spectrum and Ultra-Wideband communications. We demonstrate that, in this regime, the best one-bit quantizer does not reduce the asymptotic capacity of the coherent channel, but it does reduce that of the noncoherent channel.
\renewcommand{\thefootnote}{}
  \footnote{T.~Koch has received funding
    from the European Community's Seventh Framework Programme
    (FP7/2007-2013) under grant agreement No. 252663.}
\end{abstract}

\section{Introduction}
\label{sec:intro}
We study the effect on channel capacity of quantizing the output of the discrete-time average-power-limited Rayleigh-fading channel using a one-bit quantizer. This problem arises in communication systems where the receiver uses digital signal processing techniques, which require the analog received signal to be quantized using an analog-to-digital converter (ADC). The effects of quantization are particularly pronounced when high-resolution ADCs are not practical and low-resolution ADCs must be used \cite{walden99}.

We focus on the low signal-to-noise ratio (SNR) regime, where communication at very low spectral efficiencies takes place (as in Spread-Spectrum and Ultra-Wideband communications). For the average-power-limited \emph{real-valued Gaussian channel}, it is well-known that, in this regime, a \emph{symmetric one-bit quantizer} (which produces $1$ if the channel output is nonnegative and $0$ otherwise) reduces the capacity by a factor of $2/\pi$, corresponding to a 2dB power loss \cite{viterbiomura79}. It was recently shown that, by allowing for \emph{asymmetric one-bit quantizers} with corresponding \emph{asymmetric signal constellations}, these two decibels can be recovered in full \cite{kochlapidoth11_1}. A similar result was shown for the average-power-limited \emph{complex-valued Gaussian channel} \cite{zhangwillemshuang11}: using binary on-off keying and a \emph{radial quantizer} (which produces $1$ if the magnitude of the channel output is above some threshold and $0$ otherwise), one can achieve the low-SNR asymptotic capacity of the unquantized channel by judiciously choosing the threshold and the on-level as functions of the SNR.
Here we extend \cite{kochlapidoth11_1,zhangwillemshuang11} to \emph{Rayleigh-fading channels}. Specifically, we study the capacity per unit-energy \cite{verdu90} of such channels when the channel output is quantized using a one-bit quantizer. 

For \emph{coherent} fading channels, where the receiver has perfect channel knowledge, we show that quantizing the channel output with a one-bit quantizer causes no loss in the capacity per unit-energy. As in \cite{zhangwillemshuang11}, the capacity per unit-energy can be achieved using binary on-off keying and a radial quantizer by choosing the threshold as a function of the SNR and the fading, with the threshold and the on-level both tending to infinity as the SNR tends to zero. This result might mislead one to think that quantizing the channel output with a one-bit quantizer causes no loss in the capacity per unit-energy also for \emph{noncoherent} fading channels, where the receiver does not have perfect channel knowledge. Indeed, in the absence of a quantizer the capacity per unit-energy does not depend on whether the receiver has perfect channel knowledge or not \cite{verdu02,lapidothshamai02}. Since this capacity per unit-energy can be achieved using binary on-off keying with diverging on-level, it might therefore seem plausible that also in the presence of a quantizer the capacity per unit-energy would not depend on whether the receiver has perfect channel knowledge or not. But this is not the case: in contrast to the coherent case, quantizing the output of the \emph{noncoherent} Rayleigh-fading channel with a one-bit quantizer reduces the capacity per unit-energy.

The rest of the paper is organized as follows. Section~\ref{sec:channel} describes the channel model and introduces the capacity per unit-energy. Section~\ref{sec:1bit} presents the main results. Section~\ref{sec:2bit} discusses the capacity per unit-energy when the real and the imaginary part of the channel output are quantized separately with one-bit quantizers. And Section~\ref{sec:proof} presents the proofs of the main results.

\section{Channel Model and Capacity per Unit-Energy}
\label{sec:channel}
We consider a discrete-time Rayleigh-fading channel whose complex-valued output $\tilde{Y}_k$ at time $k\in\Integers$ corresponding to the channel input $x_k\in\Complex$ (where $\Complex$ and $\Integers$ denote the set of complex numbers and the set of integers) is given by
\begin{equation}
\label{eq:channel}
\tilde{Y}_k = H_k x_k + Z_k, \quad k\in\Integers.
\end{equation}
Here $\{Z_k,\,k\in\Integers\}$ and $\{H_k,\,k\in\Integers\}$ are independent sequences of independent and identically distributed (i.i.d.), zero-mean, circularly-symmetric, complex Gaussian random variables, the former with unit variance and the latter with variance $\sigma^2$. 
We say that the channel is \emph{coherent} if the receiver is cognizant of the realization of $\{H_k,\,k\in\Integers\}$ and that it is \emph{noncoherent} if the receiver is cognizant only of the statistics of $\{H_k,\,k\in\Integers\}$.

The receiver does not have access to the channel outputs $\{\tilde{Y}_k,\,k\in\Integers\}$ but only to a quantized version thereof. Specifically, the complex channel output $\tilde{Y}_k$ is fed to a one-bit quantizer which produces $Y_k=1$ if $\tilde{Y}_k$ is in the quantization region $\set{D}$ and $Y_k=0$ otherwise, for some Borel set $\set{D}\subset\Complex$. In the coherent case, $\set{D}$ may depend on the fading $\{H_k,\,k\in\Integers\}$.

We assume that the average power of the channel inputs is limited by $\const{P}$. The capacity of the above channel is \cite{gallager68}, \cite{biglieriproakisshamai98}
\begin{IEEEeqnarray}{lCll}
C(\const{P}) & = & \sup I(X;Y|H), \quad & \textnormal{coherent case} \label{eq:capacity_coh}\\
C(\const{P})&  = & \sup I(X;Y), \quad & \textnormal{noncoherent case} \label{eq:capacity}
\end{IEEEeqnarray}
where the suprema on the right-hand side (RHS) of \eqref{eq:capacity_coh} and \eqref{eq:capacity} are over all distributions on $X$ satisfying $\E{|X|^2}\leq\const{P}$ and over all quantization regions $\set{D}$. (Since the above channel is memoryless, we omit the time indices.)

The capacity per unit-energy is given by \cite[Th.~2]{verdu90}
\begin{equation}
\dot{C}(0) = \sup_{\const{P}>0} \frac{C(\const{P})}{\const{P}}. \label{eq:Cdotsup}
\end{equation}
It can be shown that
\begin{equation}
\dot{C}(0) = \lim_{\const{P}\downarrow 0} \frac{C(\const{P})}{\const{P}}.
\end{equation}
Thus, the capacity per unit-energy is equal to the slope at zero of the capacity-vs-power curve. It can be further shown that \cite[Th.~3]{verdu90} (see also \cite{verdu02})
\begin{equation}
\label{eq:CUC_coh}
\dot{C}(0) = \sup_{\xi\neq 0, \set{D}} \frac{D\bigl(P_{Y|H,X=\xi}\bigm\| P_{Y|H,X=0}\bigm| P_H\bigr)}{|\xi|^2}
\end{equation}
in the coherent case and
\begin{equation}
\label{eq:CUC_complexpr}
\dot{C}(0) = \sup_{\xi\neq 0, \set{D}} \frac{D\bigl(P_{Y|X=\xi}\bigm\| P_{Y|X=0}\bigr)}{|\xi|^2}
\end{equation}
in the noncoherent case. Here $D(\cdot\|\cdot)$ denotes relative entropy
\begin{equation*}
D(P\|Q) = \left\{\begin{array}{ll}\displaystyle \int \log\biggl(\frac{\d P}{\d Q}\biggr) \d P, & \textnormal{if $P \ll Q$}\\[5pt]
\infty, & \textnormal{otherwise} \end{array} \right.
\end{equation*}
(where $P\ll Q$ indicates that $P$ is absolutely continuous with respect to $Q$); $D(\cdot\|\cdot|\cdot)$ denotes conditional relative entropy
\begin{IEEEeqnarray}{lCl}
\IEEEeqnarraymulticol{3}{l}{D\bigl(P_{Y|H,X=\xi}\bigm\| P_{Y|H,X=0}\bigm| P_H\bigr)} \nonumber\\
\qquad  & = &  \int D\bigl(P_{Y|H=h,X=\xi}\bigm\| P_{Y|H=h,X=0}\bigr) \d P_H(h); \nonumber
\end{IEEEeqnarray}
$P_H$ denotes the distribution of the fading $H$; $P_{Y|X=x}$ denotes the output distribution given that the input is $x$; and $P_{Y|H=h,X=x}$ denotes the output distribution conditioned on $(H,X)=(h,x)$.

By the Data Processing Inequality \cite[Th.~2.8.1]{coverthomas91}, the capacity per unit-energy of the quantized channel is upper-bounded by that of the unquantized channel \cite{lapidothshamai02,verdu02}
\begin{equation}
\label{eq:CUC_unquant}
\dot{C}(0) \leq \frac{1}{\sigma^2}.
\end{equation}
We show that in the coherent case this upper bound holds with equality, while in the noncoherent case it is strict.

\section{Main Result}
\label{sec:1bit}
We restrict ourselves to \emph{radial} quantizers, for which
\begin{equation}
\label{eq:D_circsym}
\set{D} = \bigl\{\tilde{y}\in\Complex\colon |\tilde{y}|\geq\const{T}\bigr\}, \quad \textnormal{for some $\const{T}>0$}.
\end{equation}
In the noncoherent case---as we show in Section~\ref{sub:noncoherent}---such quantizers are optimal in the sense that they maximize the relative entropy on the RHS of \eqref{eq:CUC_complexpr} for every $\xi\neq 0$. In the coherent case such quantizers need not be optimal in the above sense. However, they suffice to achieve the capacity per unit-energy. And such quantizers have the practical advantage of not requiring knowledge of the phase of $\tilde{y}$.


\begin{theorem}
\label{theorem}
Consider the above channel model, and assume that the channel output is quantized using a one-bit quantizer.
\begin{enumerate}
\item \label{thm:coherent} In the \emph{coherent case},
\begin{equation}
\dot{C}(0) = \frac{1}{\sigma^2}
\end{equation}
which can be achieved by some radial quantizer \eqref{eq:D_circsym} with $\const{T}$ depending on $H$ and $\xi$.\footnote{Here and throughout this paper, $\xi$ refers to the parameter in \eqref{eq:CUC_coh} or \eqref{eq:CUC_complexpr}.}
\item \label{thm:noncoherent}
In the \emph{noncoherent case},
\begin{equation}
\dot{C}(0) < \frac{1}{\sigma^2}
\end{equation}
with the inequality being strict.
\end{enumerate}
\end{theorem}
\begin{proof}
See Section~\ref{sec:proof}.
\end{proof}

\section{Quantizing the Real and Imaginary Part}
\label{sec:2bit}
Instead of quantizing $\tilde{Y}$ using a one-bit quantizer, often the real and imaginary parts of $\tilde{Y}$ are quantized separately using a one-bit quantizer for each; see, e.g., \nocite{mezghaninossek07}\nocite{mezghaninossek08}\nocite{mezghaninossek09}\nocite{kronefettweis10} \cite{mezghaninossek07}--\cite{kronefettweis10}, \cite{verdu02}. Thus, the first quantizer produces $Y_{\text{R},k}=1$ if $\bigRe{\tilde{Y}_k}\in\set{D}_{\text{R}}$ and $Y_{\text{R},k}=0$ otherwise, and the second quantizer produces $Y_{\text{I},k}=1$ if $\bigIm{\tilde{Y}_k}\in\set{D}_{\text{I}}$ and $Y_{\text{I},k}=0$ otherwise, for some Borel sets $\set{D}_{\text{R}},\set{D}_{\text{I}}\subset\Reals$. (Here $\Reals$ denotes the set of real numbers, $\Re{\cdot}$ denotes the real part, and $\Im{\cdot}$ denotes the imaginary part.) In the coherent case, $\set{D}_{\text{R}}$ and $\set{D}_{\text{I}}$ may depend on the fading $\{H_k,\,k\in\Integers\}$.

The capacity per unit-energy of this channel is given by \eqref{eq:CUC_coh} or \eqref{eq:CUC_complexpr}, but with $Y$ replaced by $\bigl(Y_{\text{R}},Y_{\text{I}}\bigr)$, and with $\set{D}\subset\Complex$ replaced by $\bigl(\set{D}_{\text{R}},\set{D}_{\text{I}}\bigr)\subset\Reals\times\Reals$.

For symmetric quantizers, i.e., for
\begin{equation}
\set{D}_{\text{R}}  = \set{D}_{\text{I}} = \{u\in\Reals\colon u\geq 0\} \label{eq:sym}
\end{equation}
it follows from \cite{mezghaninossek07} and \cite[Th.~2]{singhdabeermadhow09_2} that, in the coherent case,
\begin{equation}
\label{eq:sym_coh}
\dot{C}_{\text{sym}}(0) = \frac{2}{\pi \sigma^2}.
\end{equation}
In the noncoherent case, symmetric quantizers result in zero capacity and hence, by \eqref{eq:Cdotsup}, in zero capacity per unit-energy. Indeed, for \eqref{eq:sym}
\begin{equation*}
\Prob\bigl(Y_{\text{R}}=1\bigm| X=x\bigr)=\Prob\bigl(Y_{\text{I}}=1\bigm| X=x\bigr)=\frac{1}{2}, \quad x\in\Complex.
\end{equation*}
Since, conditioned on $X$, the random variables $Y_{\text{R}}$ and $Y_{\text{I}}$ are independent, this implies that the capacity is zero. Thus, quantizing the real and imaginary parts of the Rayleigh-fading channel using symmetric one-bit quantizers reduces the capacity per unit-energy by a factor of $2/\pi$ in the coherent case, and it reduces it to zero in the noncoherent case. In the following, we show that if we allow for \emph{asymmetric} quantizers, then we can fully recover the loss of $2/\pi$ incurred on the coherent Rayleigh-fading channel. For the noncoherent channel, we show that asymmetric quantizers achieve a positive capacity per unit-energy, albeit strictly smaller than $1/\sigma^2$.

\begin{theorem}
\label{thm:2bit_coh}
Consider the above channel model, and assume that the real and imaginary parts of $\tilde{Y}$ are quantized separately using a one-bit quantizer for each.
\begin{enumerate}
\item In the \emph{coherent case},
\begin{equation}
\dot{C}(0) = \frac{1}{\sigma^2}
\end{equation}
which can be achieved by some quantization regions
\begin{IEEEeqnarray}{lCl}
\set{D}_{\text{R}}^{\star} & = & \{u \in\Reals\colon u\geq\const{T}_{\text{R}}\}\label{eq:2bit_DR} \\
\set{D}_{\text{I}}^{\star} & = & \{u\in\Reals\colon u\geq\const{T}_{\text{I}}\} \label{eq:2bit_DI}
\end{IEEEeqnarray}
where $\const{T}_{\text{R}}$ and $\const{T}_{\text{I}}$ depend on $\Re{H\xi}$ and $\Im{H\xi}$, respectively.
\item In the \emph{noncoherent case},
\begin{equation}
\frac{2\,Q(1)}{\sigma^2}\leq \dot{C}(0) < \frac{1}{\sigma^2}
\end{equation}
with the upper bound being strict. Here $Q(\cdot)$ denotes the Gaussian $Q$-function \cite[Eq.~(1.3)]{simon02}. The lower bound can be achieved by the quantization regions \eqref{eq:2bit_DR} and \eqref{eq:2bit_DI} with $\const{T}_{\text{R}}=\const{T}_{\text{I}}=(|\xi|^2+\sigma^2)/2$.
\end{enumerate}
\end{theorem}
\begin{proof}
Omitted.
\end{proof}

\section{Proof of Theorem~\ref{theorem}}
\label{sec:proof}

\subsection{Part~\ref{thm:coherent})}
\label{sub:coherent}
We show that a radial quantizer \eqref{eq:D_circsym} achieves the rate per unit-energy $1/\sigma^2$. Together with \eqref{eq:CUC_unquant}, this proves Theorem~\ref{thm:coherent}.

To this end, we first note that, conditioned on $(H,X)=(h,x)$, the squared magnitude of $\frac{2}{\sigma^2}\tilde{Y}$ is a noncentral chi-square distribution with degree $2$ and noncentrality parameter $\frac{2}{\sigma^2} |h|^2 |x|^2$ \cite[p.~8]{simon02}. Consequently, a radial quantizer yields \cite[Sec.~2-E]{simon02}
\begin{equation*}
\Prob\bigl(Y=1\bigm|H=h,X=x\bigr) = Q_1\Biggl(\!\sqrt{\frac{2}{\sigma^2}}|h| |x|, \sqrt{\frac{2}{\sigma^2}}\const{T}\!\Biggr)
\end{equation*}
where $Q_1(\cdot,\cdot)$ denotes the first-order Marcum $Q$-function \cite[Eq.~(2.20)]{simon02}. Furthermore, for $x=0$ this becomes
\begin{equation*}
\Prob\bigl(Y=1\bigm|H=h,X=0\bigr) = e^{-\frac{\const{T}^2}{\sigma^2}}.
\end{equation*}
We thus obtain 
\begin{IEEEeqnarray}{lCl}
\IEEEeqnarraymulticol{3}{l}{D\bigl(P_{Y|H,X=\xi}\bigm\| P_{Y|H,X=0}\bigm| P_H\bigr)}\nonumber\\
\quad & = & \E{Q_1\Biggl(\!\sqrt{\frac{2}{\sigma^2}}|H||\xi|, \sqrt{\frac{2}{\sigma^2}}\const{T}\!\Biggr)\log\frac{1}{e^{-\frac{\const{T}^2}{\sigma^2}}}} \nonumber\\
& & {} + \E{\Biggl\{\!1-Q_1\Biggl(\!\sqrt{\frac{2}{\sigma^2}}|H||\xi|, \sqrt{\frac{2}{\sigma^2}}\const{T}\!\Biggr)\!\Biggr\}\log\frac{1}{1-e^{-\frac{\const{T}^2}{\sigma^2}}}} \nonumber\\
& & {} - \E{H_b\Biggl(Q_1\Biggl(\!\sqrt{\frac{2}{\sigma^2}}|H||\xi|,\sqrt{\frac{2}{\sigma^2}}\const{T}\!\Biggr)\Biggr)} \nonumber\\
& \geq & \E{Q_1\Biggl(\!\sqrt{\frac{2}{\sigma^2}}|H||\xi|, \sqrt{\frac{2}{\sigma^2}}\const{T}\!\Biggr) \frac{\const{T}^2}{\sigma^2}} - \log 2 \label{eq:proof_prop7_3}
\end{IEEEeqnarray}
where $H_b(\cdot)$ denotes the binary entropy function, i.e., 
\begin{equation*}
H_b(p) \triangleq \left\{\begin{array}{ll} p \log\frac{1}{p} + (1-p) \log\frac{1}{1-p},  & \textnormal{for $0<p<1$} \\ 0, \quad & \textnormal{for $p=0$ or $p=1$.} \end{array}\right.
\end{equation*}
Here the inequality follows because the second term in the first step is nonnegative, and because the binary entropy function is upper-bounded by $\log 2$.

We choose $\const{T} = \mu |h| |\xi|$ for some fixed $\mu\in(0,1)$ and lower-bound the RHS of \eqref{eq:proof_prop7_3} using the general lower bound on the first-order Marcum $Q$-function \cite[Sec.~C-2, Eq.~(C.24)]{simon02}
\begin{equation*}
Q_1(\alpha,\beta) \geq 1- \frac{1}{2}\biggl[e^{-\frac{(\alpha-\beta)^2}{2}}-e^{-\frac{(\alpha+\beta)^2}{2}}\biggr], \quad \alpha>\beta\geq 0.
\end{equation*}
We thus obtain for the first term on the RHS of \eqref{eq:proof_prop7_3}
\begin{IEEEeqnarray}{lCl}
\IEEEeqnarraymulticol{3}{l}{\E{Q_1\Biggl(\!\sqrt{\frac{2}{\sigma^2}}|H||\xi|, \sqrt{\frac{2}{\sigma^2}}\mu|H||\xi|\!\Biggr) \frac{\mu^2 |H|^2|\xi|^2}{\sigma^2}}}\nonumber\\
\qquad & \geq & \frac{\mu^2\E{|H|^2}|\xi|^2}{\sigma^2} \nonumber\\
& & {}  -\frac{1}{2}\E{\exp\biggl(-\frac{|H|^2|\xi|^2}{\sigma^2}(1-\mu)^2\biggr)\frac{\mu^2|H|^2|\xi|^2}{\sigma^2}} \nonumber\\
& & {} +\frac{1}{2}\E{\exp\biggl(-\frac{|H|^2|\xi|^2}{\sigma^2}(1+\mu)^2\biggr)\frac{\mu^2|H|^2|\xi|^2}{\sigma^2}} \label{eq:proof_prop7_before5}\nonumber\\
& \geq & \frac{\mu^2\E{|H|^2}|\xi|^2}{\sigma^2} - \frac{\mu^2}{2\,e\,(1-\mu)^2} \label{eq:proof_prop7_5}
\end{IEEEeqnarray}
where the last step follows because $0\leq x e^{-\alpha x}\leq 1/(e\alpha)$ for every $x\geq 0$ and $\alpha>0$.

Combining \eqref{eq:proof_prop7_5} with \eqref{eq:proof_prop7_3}, and computing its ratio to $|\xi|^2$ in the limit as $|\xi|^2$ tends to infinity, yields
\begin{equation}
\dot{C}(0) \geq \frac{\mu^2\E{|H|^2}}{\sigma^2} = \frac{\mu^2}{\sigma^2}.
\end{equation}
Theorem~\ref{thm:coherent} follows then by letting $\mu$ tend to one.

\subsection{Part~\ref{thm:noncoherent})}
\label{sub:noncoherent}
We first note that, by the Data Processing Inequality for relative entropy \cite[Sec.~2.9]{coverthomas91}, the relative entropy on the RHS of \eqref{eq:CUC_complexpr} is upper-bounded by the relative entropy corresponding to the unquantized channel, i.e., \cite[Eq.~(64)]{verdu02}
\begin{equation}
\frac{D\bigl(P_{Y|X=\xi}\bigm\| P_{Y|X=0}\bigr)}{|\xi|^2} \leq \frac{1}{\sigma^2} - \frac{\log\Bigl(1+\frac{|\xi|^2}{\sigma^2}\Bigr)}{|\xi|^2}. \label{eq:proof_prop8_1}
\end{equation}
Consequently, the capacity per unit-cost \eqref{eq:CUC_complexpr} is strictly smaller than $1/\sigma^2$ unless $|\xi|$ tends to infinity. It thus remains to show that
\begin{equation}
\label{eq:proof_prop8_claim}
\varlimsup_{|\xi|\to\infty} \sup_{\set{D}} \frac{D\bigl(P_{Y|X=\xi}\bigm\| P_{Y|X=0}\bigr)}{|\xi|^2} < \frac{1}{\sigma^2}.
\end{equation}
To this end, we first note that, for every $\xi\neq 0$, the supremum in \eqref{eq:proof_prop8_claim} over all quantization regions $\set{D}$ can be replaced with the supremum over all radial quantizers \eqref{eq:D_circsym}. Indeed, for every quantization region satisfying
\begin{equation*}
\Prob\bigl(Y=1\bigm|X=\xi\bigr) = \beta, \quad 0<\beta< 1
\end{equation*}
the relative entropy
\begin{IEEEeqnarray}{lCl}
D\bigl(P_{Y|X=\xi}\bigm\| P_{Y|X=0}\bigr) & = & \beta\log\frac{1}{\Prob\bigl(Y=1\bigm| X=0\bigr)} \IEEEeqnarraynumspace\nonumber\\
\IEEEeqnarraymulticol{3}{r}{{} + (1-\beta)\log\frac{1}{1-\Prob\bigl(Y=1|X=0\bigr)} - H_b(\beta)\IEEEeqnarraynumspace}\label{eq:proof_prop8_D_beta}
\end{IEEEeqnarray}
is a convex function of $\Prob\bigl(Y=1\bigm|X=0\bigr)$. Thus, for every $0<\beta<1$, the RHS of \eqref{eq:proof_prop8_D_beta} is maximized for the quantization region that minimizes (or maximizes) $\Prob\bigl(Y=1\bigm| X=0\bigr)$ while holding $\Prob\bigl(Y=1\bigm|X=\xi\bigr)=\beta$ fixed. By the Neyman-Pearson lemma \cite{neymanpearson32}, such a quantization region has the form
\begin{equation}
\label{eq:proof_prop8_NP_D_1}
\set{D}_{\star} = \biggl\{\tilde{y}\in\Complex\colon \frac{f(\tilde{y}|0)}{f(\tilde{y}|\xi)}\leq \Lambda \biggr\}, \quad \Lambda>0
\end{equation}
(or the complement thereof), where $f(\tilde{y}|x)$ denotes the conditional density of $\tilde{Y}$, conditioned on $X=x$, and where $\Lambda$ is such that $\Prob\bigl(\tilde{Y}\in\set{D}_{\star}\bigm|X=\xi\bigr)=\beta$.
(Note that for every $0<\beta< 1$ there exists such a $\Lambda$ since, for the above channel model \eqref{eq:channel}, $\Prob\bigl(\tilde{Y}\in\set{D}_{\star}\bigm|X=\xi\bigr)$ is a continuous, strictly increasing function of $\Lambda>0$.) 
The likelihood ratio on the RHS of \eqref{eq:proof_prop8_NP_D_1} is readily evaluated as
\begin{equation*}
\frac{f(\tilde{y}|0)}{f(\tilde{y}|\xi)} = \biggl(1+\frac{|\xi|^2}{\sigma^2}\biggr) e^{-\frac{|\tilde{y}|^2}{\sigma^2}\frac{|\xi|^2}{\sigma^2+|\xi|^2}}, \quad \tilde{y}\in\Complex.
\end{equation*}
Consequently, $\set{D}_{\star}$ is the same as \eqref{eq:D_circsym} with
\begin{equation*}
\const{T} = \sigma\sqrt{\biggl(1+\frac{\sigma^2}{|\xi|^2}\biggr)\log\left(\frac{1+\frac{|\xi|^2}{\sigma^2}}{\Lambda}\right)}.
\end{equation*}
We thus obtain that, for every $\xi\neq 0$, the relative entropy $D(P_{Y|X=\xi}\| P_{Y|X=0})$ is maximized by a radial quantizer \eqref{eq:D_circsym}.

For such a quantizer, we have
\begin{equation*}
\Prob\bigl(Y=1\bigm| X=x\bigr) = \exp\biggl(-\frac{\const{T}^2}{|x|^2+\sigma^2}\biggr)\end{equation*}
which yields
\begin{IEEEeqnarray}{lCl}
\IEEEeqnarraymulticol{3}{l}{D\bigl(P_{Y|X=\xi}\bigm\| P_{Y|X=0}\bigr)}\nonumber\\
\quad & = & e^{-\frac{\const{T}^2}{|\xi|^2+\sigma^2}} \log\frac{1}{e^{-\frac{\const{T}^2}{\sigma^2}}}\nonumber\\
& & {} + \biggl[1-e^{-\frac{\const{T}^2}{|\xi|^2+\sigma^2}}\biggr] \log\frac{1}{1-e^{-\frac{\const{T}^2}{\sigma^2}}}  - H_b\biggl(e^{-\frac{\const{T}^2}{|\xi|^2+\sigma^2}}\biggr)\nonumber\\
& \leq &  \frac{\const{T}^2}{\sigma^2} e^{-\frac{\const{T}^2}{|\xi|^2+\sigma^2}} - \biggl[1-e^{-\frac{\const{T}^2}{\sigma^2}}\biggr] \log\biggl(1-e^{-\frac{\const{T}^2}{\sigma^2}}\biggr) \nonumber\\
& \leq & \frac{\const{T}^2}{\sigma^2} e^{-\frac{\const{T}^2}{|\xi|^2+\sigma^2}} + \frac{1}{e} \label{eq:proof_prop8_2}
\end{IEEEeqnarray}
where the second step follows because $H_b(\cdot)\geq 0$ and $\exp\bigl(-\const{T}^2/(|\xi|^2+\sigma^2)\bigr) \geq \exp\bigl(-\const{T}^2/\sigma^2\bigr)$; and the third step follows because $-x\log x\leq \frac{1}{e}$, $0<x<1$.

The first term on the RHS of \eqref{eq:proof_prop8_2} is maximized for \mbox{$\const{T}^2=|\xi|^2+\sigma^2$}. The RHS of \eqref{eq:proof_prop8_2} is thus upper-bounded by
\begin{equation}
\label{eq:proof_prop8_3}
D\bigl(P_{Y|X=\xi}\bigm\| P_{Y|X=0}\bigr) \leq \frac{|\xi|^2}{e\, \sigma^2} + \frac{2}{e}.
\end{equation}
Dividing the RHS of \eqref{eq:proof_prop8_3} by $|\xi|^2$, and computing the limit as $|\xi|$ tends to infinity, yields
\begin{equation}
\varlimsup_{|\xi|\to\infty}\sup_{\set{D}} \frac{D\bigl(P_{Y|X=\xi}\bigm\| P_{Y|X=0}\bigr)}{|\xi|^2} \leq \frac{1}{e\,\sigma^2} < \frac{1}{\sigma^2}.
\end{equation}
This proves Theorem~\ref{thm:noncoherent}.

\section*{Acknowledgment}
Enlightening discussions with A.~Martinez and comments from S.~Verd\'u are gratefully acknowledged.




\end{document}